\documentclass[%
reprint,
nofootinbib,
aps,
pra,
showkeys
]{revtex4-2}
\pdfoutput=1

\usepackage{amsthm, amsmath, amssymb, amsfonts, dsfont}
\usepackage{tensor, braket, physics} 

\usepackage{orcidlink}
\graphicspath{ {./images/} }

\usepackage{graphicx, placeins, float}
\usepackage[caption=false]{subfig}
\usepackage[export]{adjustbox}

\usepackage{hyperref, xcolor}
\usepackage[nameinlink,capitalize]{cleveref}
\hypersetup{
	colorlinks = true,
	linkbordercolor = {white},
	linkcolor={purple},
	citecolor={purple},
	urlcolor={blue}
}

\newcommand{\Unit}{\mathds{1}}
\newcommand{\dbra}[1]{\langle\bra{#1}}
\newcommand{\dket}[1]{\ket{#1}\rangle}

\newtheorem{theorem}{Theorem}
\newtheorem{proposition}{Proposition}
\newtheorem*{corollary*}{Properties}

\begin{document}
\title{Pauli transfer matrix direct reconstruction: channel characterization without full process tomography}
\author{Simone Roncallo\,\orcidlink{0000-0003-3506-9027}}
	\email[Simone Roncallo: ]{simone.roncallo01@ateneopv.it}
	\affiliation{Dipartimento di Fisica, Università degli Studi di Pavia, Via Agostino Bassi 6, I-27100, Pavia, Italy}
	\affiliation{INFN Sezione di Pavia, Via Agostino Bassi 6, I-27100, Pavia, Italy}
	
\author{Lorenzo Maccone\,\orcidlink{0000-0002-6729-5312}}
	\email[Lorenzo Maccone: ]{lorenzo.maccone@unipv.it}
	\affiliation{Dipartimento di Fisica, Università degli Studi di Pavia, Via Agostino Bassi 6, I-27100, Pavia, Italy}
	\affiliation{INFN Sezione di Pavia, Via Agostino Bassi 6, I-27100, Pavia, Italy}
	
\author{Chiara Macchiavello\,\orcidlink{0000-0002-2955-8759}}
	\email[Chiara Macchiavello: ]{chiara.macchiavello@unipv.it}
	\affiliation{Dipartimento di Fisica, Università degli Studi di Pavia, Via Agostino Bassi 6, I-27100, Pavia, Italy}
	\affiliation{INFN Sezione di Pavia, Via Agostino Bassi 6, I-27100, Pavia, Italy}

\begin{abstract}
    We present a tomographic protocol for the characterization of multiqubit quantum channels. We discuss a specific class of input states, for which the set of Pauli measurements at the output of the channel directly relates to its Pauli transfer matrix components. We compare our results to those of standard quantum process tomography, showing an exponential reduction in the number of different experimental configurations required by a single matrix element extraction, while keeping the same number of shots. This paves the way for more efficient experimental implementations, whenever a selective knowledge of the Pauli transfer matrix is needed. We provide several examples and simulations.
\end{abstract}
\keywords{Quantum channels; Quantum operations; Quantum process tomography; Pauli transfer matrix; Pauli-Liouville matrix;}
\maketitle

\section{Introduction}
Quantum channels, also known as quantum operations, describe the dynamics of quantum systems that interact with their surrounding environment. Intuitively, channels behave like quantum black boxes, mapping states into other states, mathematically represented as linear completely-positive trace-preserving maps \citep{book:Nielsen}. They describe a wide variety of operations, including unitary transformations, communication and teleportation protocols \citep{art:Bennett}, or noisy processes (e.g. models like the bit-flip, the depolarizing or the amplitude damping channels \citep{book:Nielsen}, even in the presence of correlations within the system \citep{art:QUIT_CorrelatedNoise, art:QUIT_QuantumCapacity}).  

Quantum process tomography (QPT) is the identification of an unknown quantum channel, obtained by correlating a complete set of input states to a complete set of measurements at their output. We call \emph{experimental configurations}, or configurations, each couple of choices of input state and output measurement. For each configuration one must feed each state and perform the same measurement for a certain number of times, called \emph{shots}, in order to retrieve the necessary statistics. QPT can be achieved in two ways \citep{art:Mohseni}: directly, when the measurement outcome immediately enters the channel reconstruction \citep{art:Greenbaum,art:Nielsen_GateTomography,art:Bendersky_SEQPT, art:Gaikwad_SQPT,art:Gaikwad_MSQPT,art:Mohseni_DCQD}(eventually after a post-processing manipulation of the data), or indirectly, i.e. when it requires additional techniques to analyse the output state (e.g. state tomography) \citep{art:Chuang_Tomography,art:Altepeter_AAQPT,art:Mataloni}.

QPT can be directly formulated in the vectorized representation (or Pauli-Liouville representation), in which operators are mapped to column vectors and channels to matrices, called \emph{Pauli transfer matrices} (PTM) \citep{art:Greenbaum, art:Nielsen_GateTomography}. This brings several advantages, e.g. the action of the channel, identified by the PTM, becomes a matrix multiplication, simplifying its inversion and manipulation in tasks like noise deconvolution \citep{art:QUIT_NoiseDeconvolution,art:QUIT_MultiQubitNoiseDeconvolution}. In this framework, the purpose of QPT is the PTM reconstruction, which is achieved by combining the outcome of different experimental configurations into each PTM entry. As in the standard Kraus description \citep{book:Nielsen}, the number of experimental configurations required by a full PTM tomography is $d^4$, where $d$ is the dimension of the system (we do not consider overcomplete sets of states, whose cost is even higher).
\begin{figure}[b]
	\centering
	\includegraphics[width = 0.35 \textwidth]{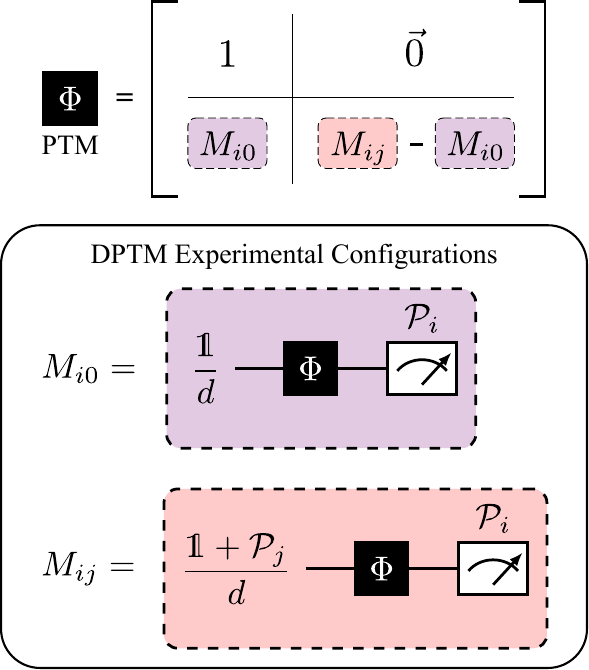}
	\caption{\label{fig:Intro}DPTM reconstruction of an unknown quantum channel $\Phi$. Two sets of experimental configurations are considered. The first one provides the characterization of the non-unital column of the channel PTM. The results in each row are then subtracted to those of the second set, completing the characterization of the remaining matrix elements. Independently from the number of qubits, the reconstruction of each PTM entry requires at most 2 experimental configurations.}
\end{figure}

In this paper we focus on ancilla-free QPT for multiqubit quantum channels, for which we propose an alternative approach that provides a direct reconstruction of the PTM (DPTM). Indeed, we consider a particular set of input states, for which each entry of the PTM is directly identified with few elements of the set of local Pauli measurements at the output of the channel, as summarized by \cref{fig:Intro}. We then compare our results to generic and to standard QPT (sQPT). We show that the DPTM reconstruction costs only $2$ experimental configurations for each PTM entry (decreasing to $1$ for unital channels) independently of the dimension of the system, providing an exponential speedup against the minimum number of configurations required by sQPT for the same task. This exponential gain is lost if one needs to reconstruct the whole PTM rather then few matrix elements. However, there are many situations where few entries are sufficient to recover the required channel characteristics, e.g. for multiparameter estimations in quantum metrology \citep{art:Giovannetti_Advances} or in assessing the unitality of a quantum channel. In contrast to other techniques like shadow tomography \citep{art:Aaronson,art:Huang-Preskill}, our approach does not improve the statistics of the reconstruction, since the total number of shots (or copies of the state) remains unchanged. Rather, it reduces the number of configurations needed, whenever only some PTM matrix elements are required or when one can introduce some prior knowledge of the channel in its characterization (examples are discussed below). Another advantage of this approach is that using less combinations of measurements can eventually reduce the systematics due to hardware errors.

Finally, we apply and simulate DPTM for two different scenarios. First, we fully characterize a single-qubit amplitude damping channel, using only $4$ experimental configurations (with respect to the $8$ required by sQPT). Then, we discuss the parameters extraction of a two-qubit correlated depolarizing channel, for which DPTM requires $2$ configurations (with respect to the $15$ required by sQPT).

\section{Quantum channels and Pauli transfer matrix}
We consider the Hilbert space of a $n$-qubits system. The basis for the set of operators is
\begin{equation}
	\big\{\sigma_{\alpha_1} \otimes \sigma_{\alpha_2} \otimes ... \otimes \sigma_{\alpha_n} \ | \ \alpha_1, \alpha_2, ..., \alpha_n = 0,1,2,3 \big\} \ ,
	\label{eq:n-basis}
\end{equation}
with $\sigma_0 = \Unit_2$, $\sigma_1 = X$, $\sigma_2 = Y$ and $\sigma_3 = Z$. We write the Pauli basis in the following notation
\begin{equation}
	\big\{ \mathcal{P}_k \ | \ k = 0,1,2,3,...,d^2-1 \big\} \ ,
	\label{eq:pauli-basis}
\end{equation}
with $d = 2^n$ and $\mathcal{P}_k$ given by the generic element of \cref{eq:n-basis} in lexicographic order.

Consider a system with quantum state $\rho$. A \emph{quantum channel} (or quantum operation) is a linear completely-positive trace-preserving (CPTP) map $\rho \to \Phi(\rho)$ \citep{book:Nielsen}. There are several ways to represent a quantum channel \citep{art:Wood_ChannelRepresentations}, e.g. the \emph{Kraus representation}, in which $\Phi(\rho)$ is described by a collection of operators $\{A_i\}_{0 \leq i \leq d^2-1}$, called Kraus operators, such that
\begin{equation}
	\Phi(\rho) = \sum_{i} A_i \rho A^{\dagger}_i \ .
	\label{eq:QuantumChannelKraus}
\end{equation}
The CPTP condition implies that
\begin{equation}
	\sum_{i} A^\dagger_i A_i = \Unit \ ,
\end{equation}
with $\Unit$ the identity operator. 

We consider the PTM representation \citep{art:Greenbaum, art:Nielsen_GateTomography} (also known as Pauli-Liouville, or superoperator,  representation), which describes $\Phi$ as a $d^2 \times d^2$ matrix 
\begin{equation}
	\Gamma_{ij} = \frac{1}{d}\Tr[\mathcal{P}_i\Phi(\mathcal{P}_j)] \ .
	\label{eq:PTMDefinition}
\end{equation} 
The PTM finds a natural application in the vectorized notation \citep{art:Greenbaum, art:Temme}, where any operator $A$ is mapped to a $1 \times d^2$ column vector 
$\dket{A}$, on which quantum channel acts through standard matrix multiplication
\begin{equation}
	\Phi(\rho) = \Gamma \dket{\rho} \ .
	\label{eq:PTMAction}
\end{equation}
This representation is equipped with the Hilbert-Schmidt inner product, so that $A_i = \langle \langle i | A\rangle \rangle$ and \cref{eq:PTMDefinition} reads
\begin{equation}
	\Gamma_{ij} = \dbra{i}\Gamma \dket{j} \ ,
\end{equation}
with $\dket{i}$ denoting the vectorized Pauli basis operator $\mathcal{P}_i$. A reshuffling transformation bijectively relates the PTM to the Choi matrix, whose spectral decomposition can lead to the original Kraus representation of the channel \citep{art:Wood_ChannelRepresentations}.

By straightforward application of \cref{eq:PTMDefinition}, the following properties of the PTM hold.
\begin{proposition}\label{th:Properties}
Consider the Hilbert space of a $n$-qubit system, with dimension $d = 2^n$. Let $\Phi$ be a quantum channel and $\Gamma$ its $d^2 \times d^2$ PTM representation, whose definition is given in \cref{eq:PTMDefinition}. The CPTP condition implies that
\begin{equation}
	\Gamma_{0j} = \delta_{0j} \ .
\end{equation}
with $0 \leq i,j \leq d^2-1$ and $\delta_{ij}$ denoting the Kronecker delta. If the channel is unital, i.e. $\Phi(\Unit) = \Unit$, then also
\begin{equation}
	\Gamma_{i0} = \delta_{i0} \ .
\end{equation}
\end{proposition}
Further  simplifications hold for Pauli channels, which are defined as those CPTP maps whose Kraus operators belong to the Pauli basis only and whose PTM is diagonal \citep{art:Flammia}. This specific class of channels includes the bit-flip, the depolarizing or the dephasing noises \cite{book:Nielsen}, also in presence of correlations \citep{art:QUIT_CorrelatedNoise, art:QUIT_Pauli}. See \citep{art:QUIT_NoiseDeconvolution} for some examples of single-qubit noise models and PTM.

\section{QPT and Direct PTM reconstruction}
The goal of QPT is to reconstruct an unknown quantum channel $\Phi$ from the statistics of a collection of experimental configurations. Namely, by applying a set of measurement operators $\{E_i\}_{0\leq i \leq d^2-1}$ to different choices of input states $\{\rho_j\}_{0\leq j \leq d^2-1}$ \citep{art:Greenbaum, art:Nielsen_GateTomography}. After having collected sufficient statistics (i.e. on a large number of shots), the outcome of each configuration reads
\begin{equation}
	p_{ij} = \Tr[E_i\Phi(\rho_j)] \ .
	\label{eq:QPTBeginning}
\end{equation}
Using the completeness relation, this can be written in terms of the channel PTM as \citep{art:Greenbaum}
\begin{equation}
	p = \alpha \Gamma \beta \ ,
	\label{eq:MeasurementsPTMQPT}	
\end{equation}
with
\begin{equation}
	\alpha_{ij} = \frac{1}{d}\Tr[E_i \mathcal{P}_j] \ , \quad \beta_{ij} = \Tr[\mathcal{P}_i \rho_j] \ .
	\label{eq:AlphaTensor}
\end{equation}
We refer to $\alpha$ and $\beta$ as reconstruction matrices, which, once inverted, provide the PTM in terms of the measurement data as
\begin{equation}
	\Gamma = \alpha^{-1} p \beta^{-1} 	\ ,
	\label{eq:QPTforPTM}
\end{equation}
whenever $\alpha^{-1}$ and $\beta^{-1}$ exist. 

In practice, the PTM is often reconstructed using a tomographic fitter instead of \cref{eq:QPTforPTM}, for example a least-squares minimization or a maximum likelihood estimation \citep{art:Greenbaum, art:MLE, art:Gambetta_MLE}. We do not consider any of these methods in our analysis, instead we compare our DPTM technique (which can equally benefit from them) to the reconstruction provided by \cref{eq:QPTforPTM}.

We now discuss an alternative procedure that provides a direct PTM reconstruction (DPTM) from the experimental data.

Consider a $n$-qubit system, prepared in one of the following states
\begin{equation}
	\rho_0 = \frac{\Unit}{d} , \quad \rho_k = \frac{\Unit + \mathcal{P}_k}{d} \quad \text{for } k \neq 0 \ ,
\end{equation}
which can be compactly written as\footnote{Index summation is always made explicit.}
\begin{equation}
	\rho_k =  \frac{1}{d}\big[ (1-\delta_{0k})\Unit + \mathcal{P}_k \big]\ .
	\label{eq:ChState}
\end{equation}
These states are positive semidefinite and normalized \citep{art:QUIT_MultiQubitNoiseDeconvolution}. Moreover, they are mixed, except for $n=1$ and $k \neq 0$. See \cref{sec:Appendix} for considerations on how they can be prepared. The channel $\Phi$ evolves $\rho_k$ to
\begin{equation}
	\Phi(\rho_k) = \frac{1}{d}\big[ (1-\delta_{0k})\Phi(\Unit) + \Phi(\mathcal{P}_k) \big]\ .
	\label{eq:EvolvedState}
\end{equation}
Consider the $i$-th element of the Pauli basis $\mathcal{P}_i$. As an observable, its expectation value against $\Phi(\rho_k)$ is
\begin{equation}
	\langle\mathcal{P}_i\rangle_{\Phi(\rho_j)} = \Tr[\mathcal{P}_i \Phi(\rho_j)] \ ,
	\label{eq:PauliExpectationValue}
\end{equation}
This mathematically represents each configuration outcome (combining input states and measurements at the output of the channel), which, expanded with \cref{eq:EvolvedState}, depends on the PTM as
\begin{equation}
	\langle\mathcal{P}_i\rangle_{\Phi(\rho_j)} = (1-\delta_{0j}) \Gamma_{i0} + \Gamma_{ij} \ .
	\label{eq:PTMfromMeasure}
\end{equation}
When $j=0$ the second contribution vanishes, yielding
\begin{equation}
	\langle\mathcal{P}_i\rangle_{\Phi(\rho_0)} =  \Gamma_{i0} \ .
\end{equation}
By subtracting this term from \cref{eq:PTMfromMeasure}, we complete the PTM reconstruction as  
\begin{equation}
\Gamma_{ij} = 
\begin{cases}
	\langle\mathcal{P}_i\rangle_{\Phi(\rho_0)} \quad &\text{for } j = 0\\
	\langle\mathcal{P}_i\rangle_{\Phi(\rho_j)}-\langle\mathcal{P}_i\rangle_{\Phi(\rho_0)} \quad &\text{for } j \neq 0
\end{cases} \ .
\label{eq:FinalExperimentalPTM}
\end{equation}
More compactly, the DPTM equation reads
\begin{equation}
 	\Gamma_{ij} = \langle\mathcal{P}_i\rangle_{\Phi(\rho_j)} - (1-\delta_{0j})\langle\mathcal{P}_i\rangle_{\Phi(\rho_0)} \ .
 	\label{eq:CompactFinalExperimentalPTM}
\end{equation}
This consistently satisfies all the properties of \cref{th:Properties}. 

The equation of DPTM establishes a direct relation between $\Gamma_{ij}$ and the corresponding configuration $\langle \mathcal{P}_i \rangle_{\Phi(\rho_j)}$. When the channel is completely unknown, DPTM requires as many resources as QPT, i.e. $d^2$ input states $\{\rho_j\}$ coupled to $d^2$ Pauli measurements $\{\mathcal{P}_i\}$. However, if some prior knowledge on the channel is already available (e.g. unitality or its, eventually incomplete, Kraus representation), it is possible to drop those combinations of $\rho_j$ and $\mathcal{P}_i$ already fixed by the initial information on $\Gamma_{ij}$. In this case the number of configurations reduces and DPTM requires fewer resources than QPT (see \cref{tab:Efficieny} for some examples).
\begin{table}[H]
	\centering
	\def\arraystretch{1.5}
	\setlength\tabcolsep{5pt}
	\begin{tabular}{ c|c|c|c|c } 
		%\cline{2-5}
		& General & CPTP & Unitality & Pauli \\
		\hline
		Resources & $d^4$ & $d^2(d^2 - 1)$ & $(d^2-1)^2$ & $d^2 - 1$ \\ 
	\end{tabular}
	\caption{\label{tab:Efficieny}Example of the DPTM cost for different type of constraints, in terms of the number of experimental configurations to collect all the data and complete the reconstruction.}
\end{table}
As discussed in the next section, similar reductions apply to any partial extraction of the PTM, e.g. in the characterization of known theoretical models, in the estimation of unknown quantum parameters, or in testing the unitality of an unknown quantum channel. DPTM can efficiently solve this last task, requiring at most $d^2 - 1$ configurations. 

The relation between QPT and DPTM can be understood by  applying the set of states of \cref{eq:ChState} to \cref{eq:QPTforPTM}. With $E_{i} = \mathcal{P}_i$, then $p_{ij} = \langle \mathcal{P}_i\rangle_{\Phi(\rho_j)}$ and the reconstruction matrices read
\begin{equation}
	\alpha_{ij} = \delta_{ij} \ , \quad \beta_{ij} =  (1-\delta_{0j})\delta_{i0} +  \delta_{ij} \ ,
\end{equation}
Their inversion yields
\begin{equation}
	\alpha^{-1}_{ij} = \delta_{ij} \ , \quad \beta_{ij}^{-1} = \delta_{ij} - (1-\delta_{0j})\delta_{i0} \ .
	\label{eq:DPTMRecoMatrices}
\end{equation}
By direct substitution, this reduces \cref{eq:QPTforPTM} precisely to \cref{eq:CompactFinalExperimentalPTM}:
from a mathematical point of view, DPTM represents a particular choice of input states for QPT. Nevertheless, in the next section we show that the specific choice of states of DPTM guarantees a faster reconstruction, in terms of the number of experiments required for each PTM entry.

\section{Comparison with standard QPT}
In this section we explicitly compare DPTM with sQPT. We label with $D$ and $Q$ the quantities and results related respectively to DPTM and sQPT.

 For $n=1$, the basis of the set of operators contains four elements
\begin{equation}
	\mathcal{P}_{0} = \Unit_2 \ , \quad \mathcal{P}_{1} = X \ , \quad \mathcal{P}_{2} = Y \ , \quad \mathcal{P}_{3} = Z \ .
\end{equation}
By substitution in \cref{eq:ChState}, the set of input states reads
\begin{align}
	&\rho_0 = \frac{1}{2}\big( \ket{0}\!\bra{0} + \ket{1}\!\bra{1}\big) \ , \label{eq:SingleQDPTMMixed} \\
	&\rho_1 = \ket{+}\!\bra{+} \ ,\label{eq:SingleQDPTMMixed1} \\
	&\rho_2 = \ket{i}\!\bra{i} \label{eq:SingleQDPTMMixed2} \ , \\
	&\rho_3 = \ket{0}\!\bra{0} \label{eq:SingleQDPTMMixed3} \ .
\end{align} 
with $\sqrt{2}\ket{\pm} = \ket{0} \pm \ket{1}$ and $\sqrt{2}\ket{\pm i} = \ket{0} \pm i\ket{1}$. In this case only $\rho_0$ is mixed, we discuss its preparation in \cref{sec:Appendix}. The remaining states are pure and can be prepared as 
\begin{equation}
	\ket{\psi_1} = H\ket{0} \ , \quad 
	\ket{\psi_2} = SH\ket{0} \ , \quad
	\ket{\psi_3} = \ket{0} \ ,
\end{equation}
with $H$ and $S$ respectively the Hadamard and the phase gates \citep{book:Nielsen}. Given a state $\rho$, each Pauli measurement $\langle\mathcal{P}_i\rangle$ can be always obtained from a $Z$-measurement, i.e. by applying a unitary transformation $\rho \to U \rho U^{\dagger}$, so that
\begin{align}
	&\langle \mathcal{P}_1 \rangle = \Tr\left[ZH \rho H \right] \ , \label{eq:1QMeasurementX} \\
	&\langle \mathcal{P}_2 \rangle = \Tr\left[ ZHS^{\dagger} \rho SH\right] \ , \label{eq:1QMeasurementY} \\
	&\langle \mathcal{P}_3 \rangle = \Tr\left[ Z \rho \right] \label{eq:1QMeasurementZ} \ ,
\end{align}
where we used 
\begin{equation}
	\mathcal{P}_1 = HZH \ , \quad \mathcal{P}_2 = SH Z HS^{\dagger} \ .
\end{equation}
Each change of basis translates the respective Pauli measurements to a simple count of the $0$, $1$ occurrences in the computational basis.\footnote{For example, if $\ket{\psi_f} = c_0 \ket{0} + c_1\ket{1}$ is a pure final state, eventually obtained by applying a unitary transformation to the evolved one, then 
$\langle Z \rangle_{\rho_f} = |c_0|^2 - |c_1|^2$.}

The DPTM configurations $M^D_{ij} := \langle \mathcal{P}_i \rangle_{\Phi(\rho_j)}$ follow by coupling each input state to the set of Pauli measurements in all the possible ways, so that
\begin{equation}
\begin{split}
	&\Gamma_{i0} = M^D_{i0} \ , \\
	&\Gamma_{ij} = M^D_{ij} - M^D_{i0} \quad \text{for } j \neq 0 \ .
\end{split}
\label{eq:1QDPTM}
\end{equation}
In terms of the reconstruction matrices $\alpha_D$ and $\beta_D$ in \cref{eq:DPTMRecoMatrices}, this gives
$\alpha_D^{-1} = \Unit$ and
\begin{equation}
	\beta_D^{-1} = \begin{pmatrix}
	+1 & -1 & -1 & -1 \\
	0 & +1 & 0 & 0 \\
	0 & 0 & +1 & 0 \\
	0 & 0 & 0 & +1
	\end{pmatrix} \ .
\end{equation} 
When the channel is unital, the identification is completely direct and $\Gamma_{ij} = M^D_{ij}$. For Pauli channels, it further simplifies to $\Gamma_{kk} = M^D_{kk}$. In both cases, the uncertainty precisely comes from the measurement, without the need for any error propagation.

For comparison, we consider the standard protocol of sQPT\footnote{As an example, we use the preparation and measurement basis from \citep{book:Nielsen}, which are the same applied by default Qiskit Experiments library.} in which the set of Pauli measurements is usually coupled to the following set of states
\begin{equation}
\begin{split}
	&\rho^Q_0 = \ket{1}\!\bra{1} \ , \quad \rho^Q_1 = \ket{+}\!\bra{+} \ , \\
	&\rho^Q_2 = \ket{i}\!\bra{i} \ \ , \quad \rho^Q_3 = \ket{0}\!\bra{0} \ .
\end{split} 
\label{eq:SingleQubitsQPTInputStates}
\end{equation}
The Pauli measurements precisely matches the one of  \cref{eq:1QMeasurementX,eq:1QMeasurementY,eq:1QMeasurementZ}, with still $\alpha^{-1}_Q = \Unit$. However, the different choice of input states modifies $\beta_Q$, whose inverse now reads
\begin{equation}
	\beta^{-1}_Q = \frac{1}{2}\begin{pmatrix}
	+1 & -1 & -1 & -1 \\
	0 & +2 & 0 & 0 \\
	0 & 0 & +2 & 0 \\
	+1 & -1 & -1 & +1
	\end{pmatrix} \ .
	\label{eq:BetaMatrixsQPT}
\end{equation}
By substitution in \cref{eq:QPTforPTM}, with $M^Q_{ij} := \langle \mathcal{P}_i \rangle_{\Phi\left(\rho^Q_j\right)}$ given by \cref{eq:MeasurementsPTMQPT}, it follows
\begin{align}
	\Gamma_{i0} &= \frac{1}{2}\left(M^Q_{i0} + M^Q_{i3}\right) \ , \label{eq:Gammai0sQPT} \\
	\Gamma_{i1} &= \frac{1}{2}\left(-M^Q_{i0} + 2M^Q_{i1} - M^Q_{i3}\right) \ , \label{eq:Gammai1sQPT} \\
	\Gamma_{i2} &= \frac{1}{2}\left(-M^Q_{i0} + 2M^Q_{i2} - M^Q_{i3}\right) \ , \label{eq:Gammai2sQPT} \\
	\Gamma_{i3} &= \frac{1}{2}\left(-M^Q_{i0} + M^Q_{i3}\right) \ . \label{eq:Gammai3sQPT}
\end{align}
This implies that, for the same PTM entry, sQPT requires more experimental configurations than DPTM, and also more resources in terms of post-processing recombinations of the data into the desired outcome. For example, a full characterization of a single-qubit Pauli channel (which has diagonal PTM) costs 8 configurations to sQPT, while only 3 for DPTM.

We end this section by comparing the cost of sQPT and DPTM for the reconstruction of a single PTM entry on $n$-qubit systems. In this framework, DPTM scales always in the same way: its states are given by \cref{eq:ChState}, while its reconstruction matrix shows high sparsity
\begin{equation}
	\beta_D^{-1} = \begin{pmatrix}
	+1 & -1 & \dots & -1 \\
	0 & \ddots & 0 & 0 \\
	\vdots & 0 & \ddots & \vdots \\
	0 & 0 & \dots & +1
	\end{pmatrix} \ .
\end{equation} 
On the other hand, the $n$-qubit states for sQPT can be obtained by taking all the possible $n$-fold Kronecker products of ${\rho_i^Q}$. \cref{eq:AlphaTensor} implies that the reconstruction matrix reads $(\beta_Q)^{\otimes n}$, with inverse $(\beta_Q^{-1})^{\otimes n}$.

To compare the performance of the two methods we define $||\Gamma_{ij}||_{D}$ and $||\Gamma_{ij}||_{Q}$ as the number of experimental configurations respectively required by DPTM and sQPT.\footnote{For example, in the single-qubit case \cref{eq:Gammai1sQPT} implies that $||\Gamma_{i1}||_{Q} = 3$.} We give the following theorem, which states that DPTM performs exponentially better than sQPT in single-entry reconstructions, independently of the number of qubits.
\begin{theorem}\label{th:Theorem1}
Consider an $n$-qubit quantum system, and a quantum (unknown) channel $\Phi$. Let $\Gamma$ be its Pauli transfer matrix, reconstructed using DPTM and sQPT. Then 
\begin{equation}
	\max_{i,j}||\Gamma_{ij}||_{D} \leq \min_{i,j} ||\Gamma_{ij}||_{Q} < \max_{i,j} ||\Gamma_{ij}||_{Q} \ ,
\end{equation}
with the strict equality satisfied only for $n=1$. Indeed
\begin{equation}
||\Gamma_{ij}||_{D} = 
	\begin{cases}
		1 \quad \text{for } j = 0 \\
		2 \quad \text{for } j \neq 0
	\end{cases} \ ,
	\label{eq:nScalingDPTM}
\end{equation}
while 
\begin{align}
	\min_{i,j} ||\Gamma_{ij}||_{Q} &= 2^n \ , \label{eq:nMinScalingQPT} \\
	\max_{i,j} ||\Gamma_{ij}||_{Q} &= 3^n \ . \label{eq:nMaxScalingQPT}
\end{align}
\end{theorem}
\begin{proof}
\cref{eq:nScalingDPTM} is a direct consequence of \cref{eq:FinalExperimentalPTM}. On the other hand, the sQPT reconstruction for $n$-qubit yields
\begin{equation}
	\Gamma = M^Q(\beta^{-1}_Q)^{\otimes n} \ .
\end{equation} 
From \cref{eq:BetaMatrixsQPT}, \cref{eq:nMinScalingQPT} and \cref{eq:nMaxScalingQPT} follow by respectively counting the minimum and maximum number of non-zero entries in the $n$-fold Kronecker product $\beta^{-1}_Q$.
\end{proof}
We summarize the content of this theorem in \cref{fig:Scaling}.
\begin{figure}[t]
	\centering
	\includegraphics[width = 0.5 \textwidth]{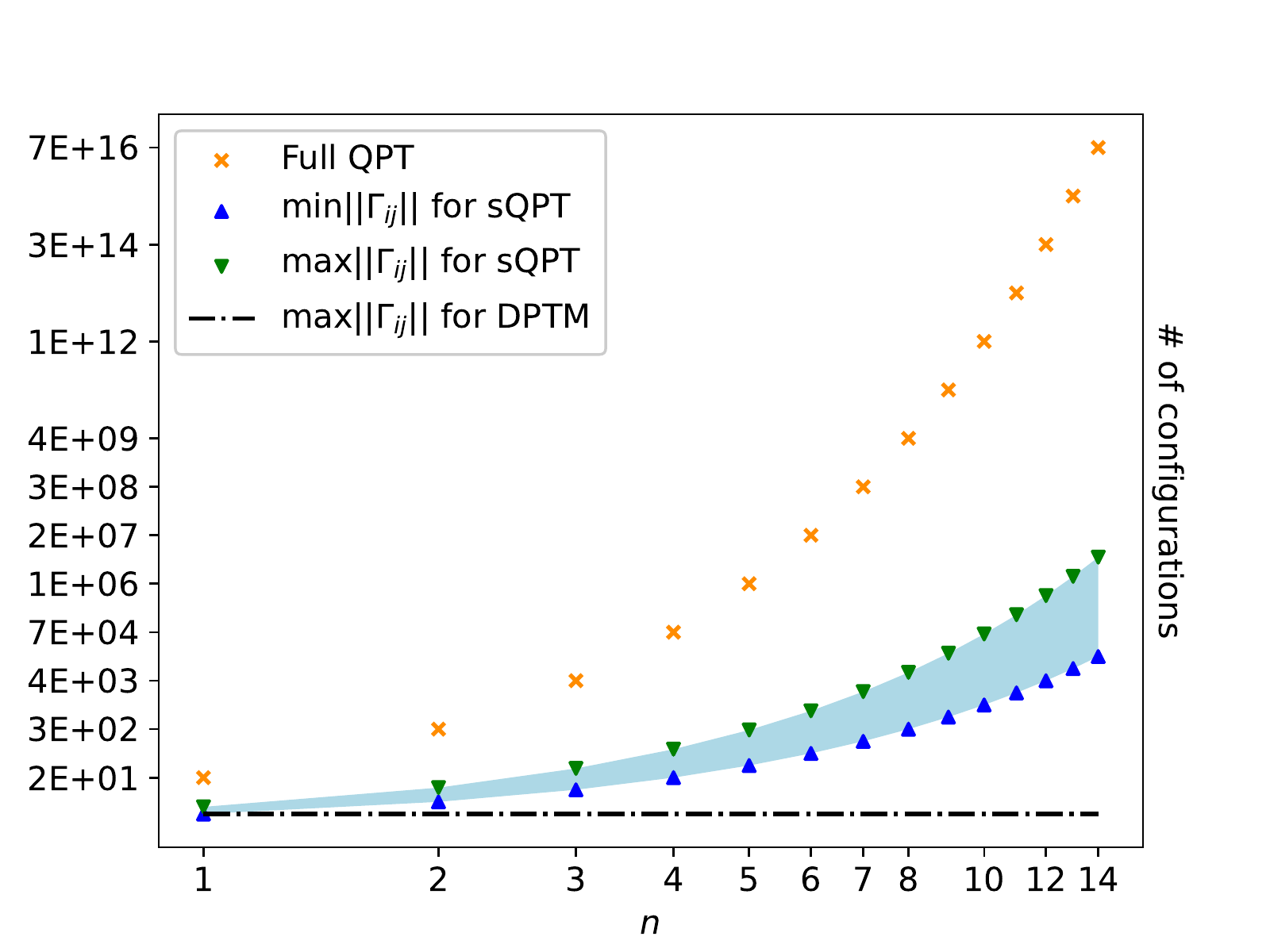}
	\caption{\label{fig:Scaling}Number of experimental configurations to reconstruct a single PTM entry, plotted with respect to the number of qubits $n$ ($log_2$ for both axes). As shown in \cref{th:Theorem1}, DPTM provides an exponential gain in terms of resources required by a single $\Gamma_{ij}$ reconstruction, improving the performance of sQPT $\forall n > 1$. For the latter, the shadowed area represents the possible configurations cost of a single $\Gamma_{ij}$ reconstruction (the exact value depends on the choice of $i$ and $j$). The results are all compared with the number of resources needed for a full process tomography (i.e. the reconstruction of the whole PTM), which requires $d^4$ configurations.}
\end{figure}
\begin{figure}[t]
	\centering
	\includegraphics[width = 0.5 \textwidth]{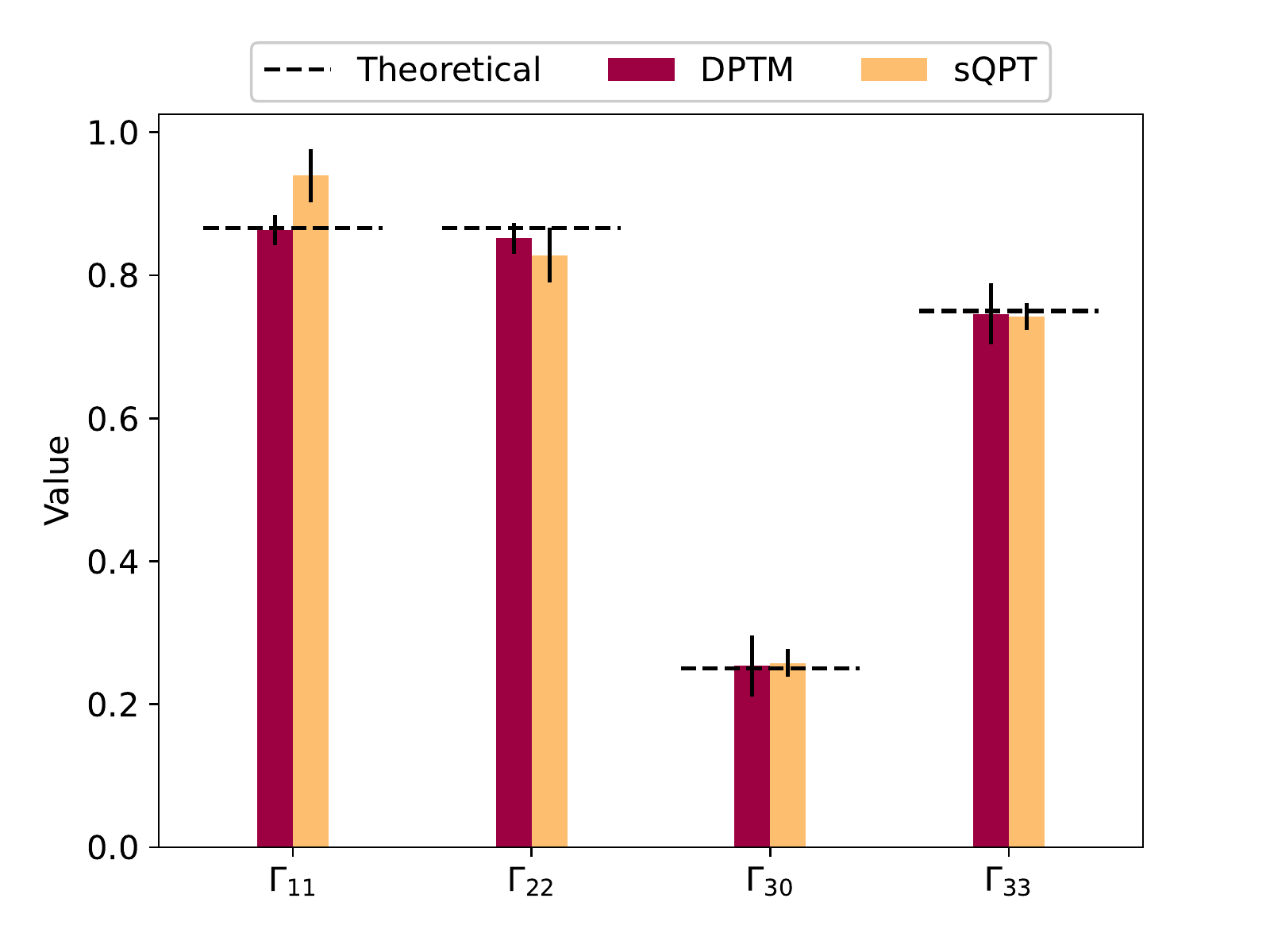}
	\caption{\label{fig:PTMComparisonQPT}Reconstruction of the non-trivial components of a single-qubit amplitude damping channel PTM. Simulated with Qiskit Aer, with transition probability $p = 0.25$. The results of DPTM and sQPT are represented respectively by the left and right bars, for each non-trivial $\Gamma_{ij}$. The dashed lines represent the theoretical values obtained from \cref{eq:PTMAmplitudeDamping}. This characterization costs $4$ experimental configurations to DPTM, and $8$ for sQPT. The number of shots for each experiment is 512. The error bars represent the standard error and are obtained by propagating the uncertainty from \cref{eq:1QDPTM} and \cref{eq:Gammai0sQPT,eq:Gammai1sQPT,eq:Gammai2sQPT,eq:Gammai3sQPT}).}
\end{figure}

\section{Applications}
In this section we provide some examples and simulations, by comparing DPTM and sQPT for a few channels.
\subsection{Amplitude damping characterization}
We consider a single-qubit amplitude damping channel with Kraus operators
\begin{equation}
	A_0 = \begin{pmatrix}
		1 & 0 \\
		0 & \sqrt{1-p}
	\end{pmatrix} \ 
	A_1 = \begin{pmatrix}
		0 & \sqrt{p} \\
		0 & 0 \
	\end{pmatrix} \ ,
	\label{eq:1QAmpDampKrausOperator}
\end{equation}
where $p$ describes the transition probability of the state $\ket{1}$ to $\ket{0}$, e.g. when the system emits a photon \citep{book:Nielsen}. The channel theoretical PTM reads
\begin{equation}
	\Gamma = \begin{pmatrix}
		1 & 0 & 0 & 0 \\
		0 & \sqrt{1-p} & 0 & 0 \\
		0 & 0 & \sqrt{1-p} & 0 \\
		p & 0 & 0 & 1 - p
	\end{pmatrix} \ .
	\label{eq:PTMAmplitudeDamping}
\end{equation}
The purpose of this section is to test both DPTM and sQPT in the characterization of the non-trivial components $\Gamma_{11}$, $\Gamma_{22}$, $\Gamma_{30}$ and $\Gamma_{33}$ (namely using \cref{eq:1QDPTM,eq:Gammai0sQPT,eq:Gammai1sQPT,eq:Gammai2sQPT,eq:Gammai3sQPT}).

In \cref{fig:PTMComparisonQPT} we plot the results of a simulation performed with Qiskit Aer for $p = 0.25$. All the results are compatible with the theoretical prediction, but DPTM reduces the cost in experimental configurations: only $4$ against the $8$ required by sQPT. We notice differences in terms of statistical uncertainties, which can be understood as follows: although sQPT combines more data into the same entries, potentially worsening the error propagation, DPTM uses mixed states, which can affect the statistics of the outcome, then leading to similar variances, in this specific example.

\subsection{Two-qubit correlated depolarizing channel}
We consider a two-qubit correlated depolarizing channel, where the amount of correlations is measured by a parameter $\mu \in [0,1]$ \citep{art:QUIT_CorrelatedNoise,art:QUIT_Pauli,art:QUIT_QuantumCapacity,art:QUIT_FullyCorrelatedDamping,art:QUIT_CorrelatedDamping}.   
We start from the class of correlated Pauli channels \citep{art:QUIT_CorrelatedNoise,art:QUIT_Pauli,art:QUIT_QuantumCapacity}, whose Kraus representation reads
\begin{equation}
	\Phi(\rho) = \sum_{\alpha_1, \alpha_2=0}^{3} p_{\alpha_1 \alpha_2} A_{\alpha_1 \alpha_2} \rho A_{\alpha_1 \alpha_2} \ ,
	\label{eq:Pauli-channel}
\end{equation}
where $A_{\alpha_1 \alpha_2} = \sigma_{\alpha_1} \otimes \sigma_{\alpha_2}$. The transition probabilities are given by the Markov chain $p_{\alpha_1 \alpha_2} = p_{\alpha_1}p_{\alpha_2|\alpha_1}$  \citep{art:QUIT_QuantumCapacity,art:Hamada}, with
\begin{align}	
	&p_{\alpha_{j}|\alpha_{i}} = (1-\mu)p_{\alpha_j} + \mu\delta_{\alpha_i \alpha_j} \, \\
	&\vec{p} = [1-p,p_x,p_y,p_z]^{T} \ . 
\end{align}
and $p = p_x + p_y + p_z$.\footnote{The case $\mu = 0$ represents a memoryless channel, that is when the qubits are completely uncorrelated, while $\mu = 1$ describes a full-memory channel, i.e. when the qubits exhibit complete correlation.} A correlated depolarizing channel is obtained by choosing $\vec{p} = [1-3p/4,p/4,p/4,p/4]^T$. 
\begin{figure}[t]
	\centering
	\subfloat[\label{fig:Gamma44}]{\includegraphics[width = 0.38 \textwidth]{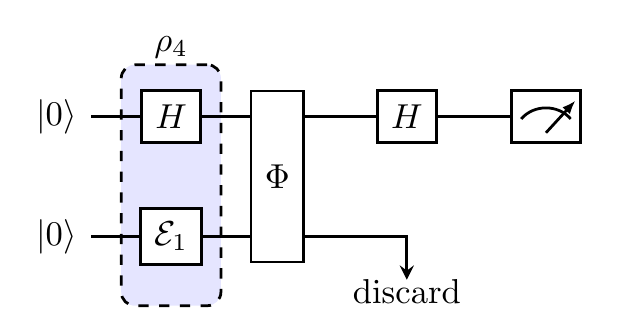}}\\[-2ex]
	\subfloat[\label{fig:Gamma66}]{\includegraphics[width = 0.425 \textwidth]{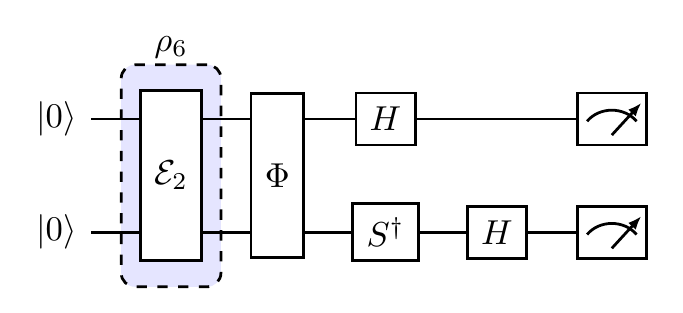}}
	\caption{\label{fig:DPTMGamma}Two-qubit circuits for the DPTM reconstruction of $\Gamma_{44}$ and $\Gamma_{66}$. (a) Circuit for $\Gamma_{44} = \langle \mathcal{P}_4 \rangle_{\Phi(\rho_4)}$, with $\mathcal{P}_4 = X \otimes I$ and $\rho_4 = [(\Unit_2 + X)\otimes \Unit_2]/4 = \ket{+}\!\bra{+}\otimes \Unit_2 /2$. The completely mixed state is prepared in the second qubit using the channel $\mathcal{E}_1(\psi) = \psi/2 + X\psi X/2$, with $\psi = \ket{0}\!\bra{0}$. (b) Circuit for $\Gamma_{66} = \langle \mathcal{P}_6 \rangle_{\Phi(\rho_6)}$, with $\mathcal{P}_6 = X \otimes Y$ and $\rho_6 = (\Unit + X\otimes Y)/4$. The input state is prepared from $\psi = \ket{0}\!\bra{0}$ using the channel $\mathcal{E}_2(\psi) = U\psi U^\dagger/2 + V U \psi U^\dagger V^\dagger/2$, with $U = \text{CNOT}(S \otimes \Unit)(H \otimes \Unit)$ and $V = (\Unit \otimes Z)(\Unit \otimes X)$.}
\end{figure}

In this section we compare DPTM and sQPT, for the case in which we only require the extraction of the parameters $p$ and $\mu$ from a set of tomographic configurations. To this extent, we compute the theoretical PTM, which yields the following relations
\begin{align}
	\Gamma_{44} & =  1 - p \ , \\
	\Gamma_{66} & = (1-p)(\mu p - p + 1) \ .
\end{align}
Their inverses allow the parameter extraction as $p \leftarrow \Gamma_{44}$ and $\mu \leftarrow \Gamma_{44}$, $\Gamma_{66}$. On one hand, these components are provided by sQPT from $15$ experimental configurations $\{M_{ij}^Q\}$, where the observables are those forming the Pauli basis and the states come from all the possible $2$-fold Kronecker products of the single-qubit set $\{\rho_i^Q\}$, with the reconstruction matrix $(\beta_Q^{-1})^{\otimes 2}$. This gives the results as
\begin{equation}
\begin{split}
	\Gamma_{44} = &- \frac{1}{4}M^Q_{40} - \frac{1}{4}M^Q_{43} + \frac{1}{2}M^Q_{44} + \\
	&+ \frac{1}{2}M^{Q}_{47} - \frac{1}{4}M^Q_{4(12)} - \frac{1}{4}M^Q_{4(15)} \ ,
\end{split}
\end{equation}
along with 
\begin{equation}
\begin{split}
	\Gamma_{66} = \frac{1}{4}M^Q_{60} - \frac{1}{2}M^Q_{62} + \frac{1}{4}M^Q_{63} - \frac{1}{2}M^Q_{64} + M^Q_{66} + \\
	 - \frac{1}{2}M^{Q}_{67} + \frac{1}{4}M^Q_{6(12)} - \frac{1}{2}M^Q_{6(14)} + \frac{1}{4}M^Q_{6(15)} \ .
\end{split}
\end{equation} 
On the other hand, DPTM provides the same components but using only $2$ configurations
\begin{equation}
	\Gamma_{44} = M^D_{44} \ , \quad \Gamma_{66} = M^D_{66} \ ,
\end{equation}
with observables in the Pauli basis and input states given by \cref{eq:ChState}. The implementation of these measurements is reported in \cref{fig:DPTMGamma}.

By choosing $p = 0.25$ and $\mu = 0.75$ so that $\Gamma_{44} = 0.750$ and $\Gamma_{66} \simeq 0.703$, we simulate DPTM and sQPT. With $2048$ shots and using Qiskit Aer, we obtain
\begin{equation}
	\text{DPTM} \rightarrow
	\begin{cases}
		\Gamma_{44} =& 0.749 \pm 0.015  \\ 
		\Gamma_{66} =& 0.710 \pm 0.016 
	\end{cases} \ , 
\end{equation}
and 
\begin{equation}
	\text{sQPT} \rightarrow
	\begin{cases}
		\Gamma_{44} =& 0.756 \pm 0.015 \\ 
		\Gamma_{66} =& 0.705 \pm 0.016
	\end{cases} \ .
\end{equation}
We notice that both methods are in agreement with the theoretical prediction. In this regard, DPTM performs better than sQPT: requiring fewer experimental configurations under the same number of shots and with compatible errors. Importantly, reducing the number of required settings will typically reduce the systematics due to hardware errors (which are not considered in this simulation).
 
\section{Conclusions}
In this work we applied QPT to the reconstruction of a multiqubit quantum channel PTM. In general, QPT performs a set of measurements on different experimental configurations (i.e. by changing the input state and/or the observable at the output of the channel), combining them into each PTM entry at the post-processing stage.

We presented an alternative technique that provides a direct reconstruction of the PTM from the measurement outcomes. In principle, our approach differs from sQPT only in the choice of the input states. However, this choice simplifies both the experimental implementation and the post-processing reconstruction: DPTM exponentially reduces the number of different configurations that combines into a single PTM entry, while keeping the same number of shots of sQPT. 

Though both techniques require $d^4$ configurations for a full tomography of the channel, DPTM truly shines when only a subset of the PTM has to be reconstructed (e.g. in the extraction of the channel parameters under a given theoretical model, or for biased characterizations of unknown channels, for example of the Pauli type). While not improving the statistics of the results, DPTM requires (at most) $2$ experimental configurations for each PTM entry, independently of the dimension of the system: this allows for more efficient (and scalable) experimental implementations of tomographic protocols, with fewer computational circuits or setups of the optical table.

\section*{Acknowledgments}
This work received support from MIUR Dipartimenti di Eccellenza 2018-2022, Project No. F11I18000680001, from EU H2020 QuantERA ERA-NET Cofund in Quantum Technologies, Quantum Information and Communication with High-dimensional Encoding (QuICHE), Grant Agreement 731473 and 101017733, from the U.S. Department of Energy, Office of Science, National Quantum Information Science Research Centers, Superconducting Quantum Materials and Systems Center (SQMS), Contract No. DE-AC02-07CH11359. L.M. acknowledges support from the PNRR MUR Project PE0000023-NQSTI. C.M. acknowledges support from the National Research Centre for HPC, Big Data and Quantum Computing, PNRR MUR Project CN0000013-ICSC.

\appendix*
\section{MIXED STATE PREPARATION\label{sec:Appendix}}
In this section we address the problem of generating the set of DPTM input states of \cref{eq:ChState}. Indeed, this discussion applies to any set of input, possibly mixed, states.
 
For single-qubit channels, $\rho_0$ in \cref{eq:SingleQDPTMMixed} is the only mixed state of the set. Its implementation can be achieved either through an ancilla-assisted approach: namely by preparing the system in the maximally entangled state
\begin{equation}
	\ket{\Psi_0} = \frac{1}{\sqrt{2}}\sum_{k=0}^{1}\ket{k}\otimes\ket{k}_a \ ,
\end{equation}
and tracing out the ancillary qubit
\begin{equation}
	\rho_0 = \Tr_{a}\big[\ket{\Psi_0}\!\bra{\Psi_0}\big] \ ,
\end{equation}
or by classical random generation of its pure components
\begin{equation}
	\ket{\psi_0} \in \{ \ket{0}, \ket{1} \} \quad \text{with } p(\ket{0}) = p(\ket{1}) = \frac{1}{2} \ .
\end{equation}
In principle, this last approach is equivalent to preparing the initial state $\psi = \ket{0}\!\bra{0}$ and then applying a quantum channel 
\begin{equation}
 	\mathcal{E}_0(\psi) = \frac{1}{2}\psi + \frac{1}{2}X\psi X \ ,
\end{equation}
so that $\rho_0 = \mathcal{E}_0(\psi)$.

Starting from this last consideration, we now discuss the multiqubit generation of a set of possibly mixed states $\{\rho_k\}$ (e.g. the DPTM input states of \cref{eq:ChState}). Consider $\{\psi_k\}$ a set of $d^2$, linearly independent, \emph{pure} states. The requirement of purity allows to start from states that can be easily and procedurally generated. Let $\mathcal{E}$ be the channel that maps each initial pure state to an element of the mixed set, i.e. $\rho_k = \mathcal{E}(\psi_k) \ \forall k$. 
Indeed, we can split the preparation of a set of mixed states to that of a set of pure states and the simulation of a  quantum channel.

To determine $\mathcal{E}$ from $\{\rho_k\}$ and $\{\psi_k\}$, we move again to the vectorized notation, which yields
\begin{equation}
	\rho_k = \Lambda \psi_k \ , 
\end{equation}
with $\Lambda$ the PTM of $\mathcal{E}$. This problem is a quantum process tomography task, in which each combination of input-output states is already known. In this case, \cref{eq:QPTforPTM} reads
\begin{equation}
	\Lambda = \pi B^{-1} \ ,
	\label{eq:MixedPreparationwithTomography}
\end{equation}
where $E_i = \mathcal{P}_i$ and
\begin{align}
	B_{ij} &= \Tr[\mathcal{P}_i \psi_j] \ , \\
	\pi_{ij} &= \Tr[\mathcal{P}_i\mathcal{E}(\psi_j)] = \Tr[\mathcal{P}_i \rho_j] \ .
\end{align}
Since both $\{\rho_k\}$ and $\{\psi_k\}$ are chosen, we can theoretically compute $\Lambda$, and then procedurally simulate it using circuits methods like \texttt{PTM.to\_instruction()} from the Qiskit package.
\begin{figure}[t]
	\centering
	\includegraphics[width = 0.38 \textwidth]{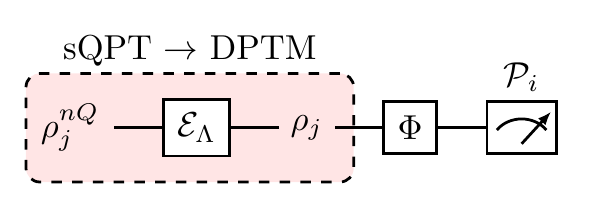}
	\caption{\label{fig:sQPTtoDPTM}Conversion between sQPT and DPTM, in the identification of an unknown quantum channel $\Phi$. The circuit corresponds to the configuration $\langle \mathcal{P}_i\rangle_{\Phi(\rho_j)}$, with the system initially prepared in the sQPT multiqubit input states $\{\rho_j^{nQ}\}$ (obtained by taking all the possible $n$-fold Kronecker products of ${\rho_i^Q}$). The conversion is implemented through the channel $\mathcal{E}$, with PTM $\Lambda$ given by \cref{eq:MixedPreparationwithTomography}.}
\end{figure} 

When the states involved are those of DPTM and sQPT, the channel $\mathcal{E}$ provides a map between these two tomographic reconstructions, with $\pi$ and $B$ respectively given by $\beta_D$ and $(\beta_Q)^{\otimes n}$. We summarize this protocol in \cref{fig:sQPTtoDPTM}.

Returning to the original example, we consider $n=1$. The channel that prepares the DPTM states $\{\rho_k\}$ (see \cref{eq:SingleQDPTMMixed,eq:SingleQDPTMMixed1,eq:SingleQDPTMMixed2,eq:SingleQDPTMMixed3}) from those of sQPT $\{\rho_k^Q\}$ (see \cref{eq:SingleQubitsQPTInputStates}) reads
\begin{equation}
	\Lambda = \begin{pmatrix}
		1 & 0 & 0 & 0 \\
		0 & 1 & 0 & 0 \\
		0 & 0 & 1 & 0 \\
		\frac{1}{2} & -\frac{1}{2} & -\frac{1}{2} & \frac{1}{2}
	\end{pmatrix} \ ,
\end{equation}
which is almost trivial, since for $n=1$ the states $\{\rho_k\}$ and $\{\rho^Q_k\}$ differ only for $k=0$.

\bibliography{refs.bib}

\begin{thebibliography}{28}%
\makeatletter
\providecommand \@ifxundefined [1]{%
 \@ifx{#1\undefined}
}%
\providecommand \@ifnum [1]{%
 \ifnum #1\expandafter \@firstoftwo
 \else \expandafter \@secondoftwo
 \fi
}%
\providecommand \@ifx [1]{%
 \ifx #1\expandafter \@firstoftwo
 \else \expandafter \@secondoftwo
 \fi
}%
\providecommand \natexlab [1]{#1}%
\providecommand \enquote  [1]{``#1''}%
\providecommand \bibnamefont  [1]{#1}%
\providecommand \bibfnamefont [1]{#1}%
\providecommand \citenamefont [1]{#1}%
\providecommand \href@noop [0]{\@secondoftwo}%
\providecommand \href [0]{\begingroup \@sanitize@url \@href}%
\providecommand \@href[1]{\@@startlink{#1}\@@href}%
\providecommand \@@href[1]{\endgroup#1\@@endlink}%
\providecommand \@sanitize@url [0]{\catcode `\\12\catcode `\$12\catcode
  `\&12\catcode `\#12\catcode `\^12\catcode `\_12\catcode `\%12\relax}%
\providecommand \@@startlink[1]{}%
\providecommand \@@endlink[0]{}%
\providecommand \url  [0]{\begingroup\@sanitize@url \@url }%
\providecommand \@url [1]{\endgroup\@href {#1}{\urlprefix }}%
\providecommand \urlprefix  [0]{URL }%
\providecommand \Eprint [0]{\href }%
\providecommand \doibase [0]{https://doi.org/}%
\providecommand \selectlanguage [0]{\@gobble}%
\providecommand \bibinfo  [0]{\@secondoftwo}%
\providecommand \bibfield  [0]{\@secondoftwo}%
\providecommand \translation [1]{[#1]}%
\providecommand \BibitemOpen [0]{}%
\providecommand \bibitemStop [0]{}%
\providecommand \bibitemNoStop [0]{.\EOS\space}%
\providecommand \EOS [0]{\spacefactor3000\relax}%
\providecommand \BibitemShut  [1]{\csname bibitem#1\endcsname}%
\let\auto@bib@innerbib\@empty
%</preamble>
\bibitem [{\citenamefont {Nielsen}\ and\ \citenamefont
  {Chuang}(2010)}]{book:Nielsen}%
  \BibitemOpen
  \bibfield  {author} {\bibinfo {author} {\bibfnamefont {M.~A.}\ \bibnamefont
  {Nielsen}}\ and\ \bibinfo {author} {\bibfnamefont {I.~L.}\ \bibnamefont
  {Chuang}},\ }\href {https://doi.org/10.1017/CBO9780511976667} {\emph
  {\bibinfo {title} {Quantum Computation and Quantum Information: 10th
  Anniversary Edition}}}\ (\bibinfo  {publisher} {Cambridge University Press},\
  \bibinfo {year} {2010})\BibitemShut {NoStop}%
\bibitem [{\citenamefont {Bennett}\ \emph {et~al.}(1993)\citenamefont
  {Bennett}, \citenamefont {Brassard}, \citenamefont {Cr\'epeau}, \citenamefont
  {Jozsa}, \citenamefont {Peres},\ and\ \citenamefont
  {Wootters}}]{art:Bennett}%
  \BibitemOpen
  \bibfield  {author} {\bibinfo {author} {\bibfnamefont {C.~H.}\ \bibnamefont
  {Bennett}}, \bibinfo {author} {\bibfnamefont {G.}~\bibnamefont {Brassard}},
  \bibinfo {author} {\bibfnamefont {C.}~\bibnamefont {Cr\'epeau}}, \bibinfo
  {author} {\bibfnamefont {R.}~\bibnamefont {Jozsa}}, \bibinfo {author}
  {\bibfnamefont {A.}~\bibnamefont {Peres}},\ and\ \bibinfo {author}
  {\bibfnamefont {W.~K.}\ \bibnamefont {Wootters}},\ }\bibfield  {title}
  {\bibinfo {title} {Teleporting an unknown quantum state via dual classical
  and {E}instein-{P}odolsky-{R}osen channels},\ }\href
  {https://doi.org/10.1103/PhysRevLett.70.1895} {\bibfield  {journal} {\bibinfo
   {journal} {Phys. Rev. Lett.}\ }\textbf {\bibinfo {volume} {70}},\ \bibinfo
  {pages} {1895} (\bibinfo {year} {1993})}\BibitemShut {NoStop}%
\bibitem [{\citenamefont {Macchiavello}\ and\ \citenamefont
  {Palma}(2002)}]{art:QUIT_CorrelatedNoise}%
  \BibitemOpen
  \bibfield  {author} {\bibinfo {author} {\bibfnamefont {C.}~\bibnamefont
  {Macchiavello}}\ and\ \bibinfo {author} {\bibfnamefont {G.~M.}\ \bibnamefont
  {Palma}},\ }\bibfield  {title} {\bibinfo {title} {Entanglement-enhanced
  information transmission over a quantum channel with correlated noise},\
  }\href {https://doi.org/10.1103/physreva.65.050301} {\bibfield  {journal}
  {\bibinfo  {journal} {Phys. Rev. A}\ }\textbf {\bibinfo {volume} {65}},\
  \bibinfo {pages} {050301(R)} (\bibinfo {year} {2002})}\BibitemShut {NoStop}%
\bibitem [{\citenamefont {Macchiavello}\ and\ \citenamefont
  {Sacchi}(2016)}]{art:QUIT_QuantumCapacity}%
  \BibitemOpen
  \bibfield  {author} {\bibinfo {author} {\bibfnamefont {C.}~\bibnamefont
  {Macchiavello}}\ and\ \bibinfo {author} {\bibfnamefont {M.~F.}\ \bibnamefont
  {Sacchi}},\ }\bibfield  {title} {\bibinfo {title} {Witnessing quantum
  capacities of correlated channels},\ }\href
  {https://doi.org/10.1103/physreva.94.052333} {\bibfield  {journal} {\bibinfo
  {journal} {Phys. Rev. A}\ }\textbf {\bibinfo {volume} {94}},\ \bibinfo
  {pages} {052333} (\bibinfo {year} {2016})}\BibitemShut {NoStop}%
\bibitem [{\citenamefont {Mohseni}\ \emph {et~al.}(2008)\citenamefont
  {Mohseni}, \citenamefont {Rezakhani},\ and\ \citenamefont
  {Lidar}}]{art:Mohseni}%
  \BibitemOpen
  \bibfield  {author} {\bibinfo {author} {\bibfnamefont {M.}~\bibnamefont
  {Mohseni}}, \bibinfo {author} {\bibfnamefont {A.~T.}\ \bibnamefont
  {Rezakhani}},\ and\ \bibinfo {author} {\bibfnamefont {D.~A.}\ \bibnamefont
  {Lidar}},\ }\bibfield  {title} {\bibinfo {title} {Quantum-process tomography:
  Resource analysis of different strategies},\ }\href
  {https://doi.org/10.1103/physreva.77.032322} {\bibfield  {journal} {\bibinfo
  {journal} {Phys. Rev. A}\ }\textbf {\bibinfo {volume} {77}},\ \bibinfo
  {pages} {032322} (\bibinfo {year} {2008})}\BibitemShut {NoStop}%
\bibitem [{\citenamefont {Greenbaum}(2015)}]{art:Greenbaum}%
  \BibitemOpen
  \bibfield  {author} {\bibinfo {author} {\bibfnamefont {D.}~\bibnamefont
  {Greenbaum}},\ }\href@noop {} {\bibinfo {title} {Introduction to quantum gate
  set tomography}} (\bibinfo {year} {2015}),\ \Eprint
  {https://arxiv.org/abs/1509.02921} {arXiv:1509.02921 [quant-ph]} \BibitemShut
  {NoStop}%
\bibitem [{\citenamefont {Nielsen}\ \emph {et~al.}(2021)\citenamefont
  {Nielsen}, \citenamefont {Gamble}, \citenamefont {Rudinger}, \citenamefont
  {Scholten}, \citenamefont {Young},\ and\ \citenamefont
  {Blume-Kohout}}]{art:Nielsen_GateTomography}%
  \BibitemOpen
  \bibfield  {author} {\bibinfo {author} {\bibfnamefont {E.}~\bibnamefont
  {Nielsen}}, \bibinfo {author} {\bibfnamefont {J.~K.}\ \bibnamefont {Gamble}},
  \bibinfo {author} {\bibfnamefont {K.}~\bibnamefont {Rudinger}}, \bibinfo
  {author} {\bibfnamefont {T.}~\bibnamefont {Scholten}}, \bibinfo {author}
  {\bibfnamefont {K.}~\bibnamefont {Young}},\ and\ \bibinfo {author}
  {\bibfnamefont {R.}~\bibnamefont {Blume-Kohout}},\ }\bibfield  {title}
  {\bibinfo {title} {Gate set tomography},\ }\href
  {https://doi.org/10.22331/q-2021-10-05-557} {\bibfield  {journal} {\bibinfo
  {journal} {Quantum}\ }\textbf {\bibinfo {volume} {5}},\ \bibinfo {pages}
  {557} (\bibinfo {year} {2021})}\BibitemShut {NoStop}%
\bibitem [{\citenamefont {Bendersky}\ \emph {et~al.}(2009)\citenamefont
  {Bendersky}, \citenamefont {Pastawski},\ and\ \citenamefont
  {Paz}}]{art:Bendersky_SEQPT}%
  \BibitemOpen
  \bibfield  {author} {\bibinfo {author} {\bibfnamefont {A.}~\bibnamefont
  {Bendersky}}, \bibinfo {author} {\bibfnamefont {F.}~\bibnamefont
  {Pastawski}},\ and\ \bibinfo {author} {\bibfnamefont {J.~P.}\ \bibnamefont
  {Paz}},\ }\bibfield  {title} {\bibinfo {title} {Selective and efficient
  quantum process tomography},\ }\href
  {https://doi.org/10.1103/PhysRevA.80.032116} {\bibfield  {journal} {\bibinfo
  {journal} {Phys. Rev. A}\ }\textbf {\bibinfo {volume} {80}},\ \bibinfo
  {pages} {032116} (\bibinfo {year} {2009})}\BibitemShut {NoStop}%
\bibitem [{\citenamefont {Gaikwad}\ \emph {et~al.}(2018)\citenamefont
  {Gaikwad}, \citenamefont {Rehal}, \citenamefont {Singh}, \citenamefont
  {Arvind},\ and\ \citenamefont {Dorai}}]{art:Gaikwad_SQPT}%
  \BibitemOpen
  \bibfield  {author} {\bibinfo {author} {\bibfnamefont {A.}~\bibnamefont
  {Gaikwad}}, \bibinfo {author} {\bibfnamefont {D.}~\bibnamefont {Rehal}},
  \bibinfo {author} {\bibfnamefont {A.}~\bibnamefont {Singh}}, \bibinfo
  {author} {\bibnamefont {Arvind}},\ and\ \bibinfo {author} {\bibfnamefont
  {K.}~\bibnamefont {Dorai}},\ }\bibfield  {title} {\bibinfo {title}
  {Experimental demonstration of selective quantum process tomography on an
  {NMR} quantum information processor},\ }\href
  {https://doi.org/10.1103/PhysRevA.97.022311} {\bibfield  {journal} {\bibinfo
  {journal} {Phys. Rev. A}\ }\textbf {\bibinfo {volume} {97}},\ \bibinfo
  {pages} {022311} (\bibinfo {year} {2018})}\BibitemShut {NoStop}%
\bibitem [{\citenamefont {Gaikwad}\ \emph {et~al.}(2022)\citenamefont
  {Gaikwad}, \citenamefont {Shende}, \citenamefont {Arvind},\ and\
  \citenamefont {Dorai}}]{art:Gaikwad_MSQPT}%
  \BibitemOpen
  \bibfield  {author} {\bibinfo {author} {\bibfnamefont {A.}~\bibnamefont
  {Gaikwad}}, \bibinfo {author} {\bibfnamefont {K.}~\bibnamefont {Shende}},
  \bibinfo {author} {\bibnamefont {Arvind}},\ and\ \bibinfo {author}
  {\bibfnamefont {K.}~\bibnamefont {Dorai}},\ }\bibfield  {title} {\bibinfo
  {title} {Implementing efficient selective quantum process tomography of
  superconducting quantum gates on {IBM} quantum experience},\ }\href
  {https://doi.org/10.1038/s41598-022-07721-3} {\bibfield  {journal} {\bibinfo
  {journal} {Sci. Rep.}\ }\textbf {\bibinfo {volume} {12}},\ \bibinfo {pages}
  {3688} (\bibinfo {year} {2022})}\BibitemShut {NoStop}%
\bibitem [{\citenamefont {Mohseni}\ and\ \citenamefont
  {Lidar}(2006)}]{art:Mohseni_DCQD}%
  \BibitemOpen
  \bibfield  {author} {\bibinfo {author} {\bibfnamefont {M.}~\bibnamefont
  {Mohseni}}\ and\ \bibinfo {author} {\bibfnamefont {D.~A.}\ \bibnamefont
  {Lidar}},\ }\bibfield  {title} {\bibinfo {title} {Direct characterization of
  quantum dynamics},\ }\href {https://doi.org/10.1103/PhysRevLett.97.170501}
  {\bibfield  {journal} {\bibinfo  {journal} {Phys. Rev. Lett.}\ }\textbf
  {\bibinfo {volume} {97}},\ \bibinfo {pages} {170501} (\bibinfo {year}
  {2006})}\BibitemShut {NoStop}%
\bibitem [{\citenamefont {Chuang}\ and\ \citenamefont
  {Nielsen}(1997)}]{art:Chuang_Tomography}%
  \BibitemOpen
  \bibfield  {author} {\bibinfo {author} {\bibfnamefont {I.~L.}\ \bibnamefont
  {Chuang}}\ and\ \bibinfo {author} {\bibfnamefont {M.~A.}\ \bibnamefont
  {Nielsen}},\ }\bibfield  {title} {\bibinfo {title} {Prescription for
  experimental determination of the dynamics of a quantum black box},\ }\href
  {https://doi.org/10.1080/09500349708231894} {\bibfield  {journal} {\bibinfo
  {journal} {J. Mod. Opt.}\ }\textbf {\bibinfo {volume} {44}},\ \bibinfo
  {pages} {2455} (\bibinfo {year} {1997})}\BibitemShut {NoStop}%
\bibitem [{\citenamefont {Altepeter}\ \emph {et~al.}(2003)\citenamefont
  {Altepeter}, \citenamefont {Branning}, \citenamefont {Jeffrey}, \citenamefont
  {Wei}, \citenamefont {Kwiat}, \citenamefont {Thew}, \citenamefont {O'Brien},
  \citenamefont {Nielsen},\ and\ \citenamefont {White}}]{art:Altepeter_AAQPT}%
  \BibitemOpen
  \bibfield  {author} {\bibinfo {author} {\bibfnamefont {J.~B.}\ \bibnamefont
  {Altepeter}}, \bibinfo {author} {\bibfnamefont {D.}~\bibnamefont {Branning}},
  \bibinfo {author} {\bibfnamefont {E.}~\bibnamefont {Jeffrey}}, \bibinfo
  {author} {\bibfnamefont {T.~C.}\ \bibnamefont {Wei}}, \bibinfo {author}
  {\bibfnamefont {P.~G.}\ \bibnamefont {Kwiat}}, \bibinfo {author}
  {\bibfnamefont {R.~T.}\ \bibnamefont {Thew}}, \bibinfo {author}
  {\bibfnamefont {J.~L.}\ \bibnamefont {O'Brien}}, \bibinfo {author}
  {\bibfnamefont {M.~A.}\ \bibnamefont {Nielsen}},\ and\ \bibinfo {author}
  {\bibfnamefont {A.~G.}\ \bibnamefont {White}},\ }\bibfield  {title} {\bibinfo
  {title} {Ancilla-assisted quantum process tomography},\ }\href
  {https://doi.org/10.1103/PhysRevLett.90.193601} {\bibfield  {journal}
  {\bibinfo  {journal} {Phys. Rev. Lett.}\ }\textbf {\bibinfo {volume} {90}},\
  \bibinfo {pages} {193601} (\bibinfo {year} {2003})}\BibitemShut {NoStop}%
\bibitem [{\citenamefont {Bongioanni}\ \emph {et~al.}(2010)\citenamefont
  {Bongioanni}, \citenamefont {Sansoni}, \citenamefont {Sciarrino},
  \citenamefont {Vallone},\ and\ \citenamefont {Mataloni}}]{art:Mataloni}%
  \BibitemOpen
  \bibfield  {author} {\bibinfo {author} {\bibfnamefont {I.}~\bibnamefont
  {Bongioanni}}, \bibinfo {author} {\bibfnamefont {L.}~\bibnamefont {Sansoni}},
  \bibinfo {author} {\bibfnamefont {F.}~\bibnamefont {Sciarrino}}, \bibinfo
  {author} {\bibfnamefont {G.}~\bibnamefont {Vallone}},\ and\ \bibinfo {author}
  {\bibfnamefont {P.}~\bibnamefont {Mataloni}},\ }\bibfield  {title} {\bibinfo
  {title} {Experimental quantum process tomography of non-trace-preserving
  maps},\ }\href {https://doi.org/10.1103/PhysRevA.82.042307} {\bibfield
  {journal} {\bibinfo  {journal} {Phys. Rev. A}\ }\textbf {\bibinfo {volume}
  {82}},\ \bibinfo {pages} {042307} (\bibinfo {year} {2010})}\BibitemShut
  {NoStop}%
\bibitem [{\citenamefont {Mangini}\ \emph {et~al.}(2022)\citenamefont
  {Mangini}, \citenamefont {Maccone},\ and\ \citenamefont
  {Macchiavello}}]{art:QUIT_NoiseDeconvolution}%
  \BibitemOpen
  \bibfield  {author} {\bibinfo {author} {\bibfnamefont {S.}~\bibnamefont
  {Mangini}}, \bibinfo {author} {\bibfnamefont {L.}~\bibnamefont {Maccone}},\
  and\ \bibinfo {author} {\bibfnamefont {C.}~\bibnamefont {Macchiavello}},\
  }\bibfield  {title} {\bibinfo {title} {Qubit noise deconvolution},\ }\href
  {https://doi.org/10.1140/epjqt/s40507-022-00151-0} {\bibfield  {journal}
  {\bibinfo  {journal} {{EPJ} Quantum Technol.}\ }\textbf {\bibinfo {volume}
  {9}},\ \bibinfo {pages} {29} (\bibinfo {year} {2022})}\BibitemShut {NoStop}%
\bibitem [{\citenamefont {Roncallo}\ \emph {et~al.}(2023)\citenamefont
  {Roncallo}, \citenamefont {Maccone},\ and\ \citenamefont
  {Macchiavello}}]{art:QUIT_MultiQubitNoiseDeconvolution}%
  \BibitemOpen
  \bibfield  {author} {\bibinfo {author} {\bibfnamefont {S.}~\bibnamefont
  {Roncallo}}, \bibinfo {author} {\bibfnamefont {L.}~\bibnamefont {Maccone}},\
  and\ \bibinfo {author} {\bibfnamefont {C.}~\bibnamefont {Macchiavello}},\
  }\bibfield  {title} {\bibinfo {title} {Multiqubit noise deconvolution and
  characterization},\ }\href {https://doi.org/10.1103/PhysRevA.107.022419}
  {\bibfield  {journal} {\bibinfo  {journal} {Phys. Rev. A}\ }\textbf {\bibinfo
  {volume} {107}},\ \bibinfo {pages} {022419} (\bibinfo {year}
  {2023})}\BibitemShut {NoStop}%
\bibitem [{\citenamefont {Giovannetti}\ \emph {et~al.}(2011)\citenamefont
  {Giovannetti}, \citenamefont {Lloyd},\ and\ \citenamefont
  {Maccone}}]{art:Giovannetti_Advances}%
  \BibitemOpen
  \bibfield  {author} {\bibinfo {author} {\bibfnamefont {V.}~\bibnamefont
  {Giovannetti}}, \bibinfo {author} {\bibfnamefont {S.}~\bibnamefont {Lloyd}},\
  and\ \bibinfo {author} {\bibfnamefont {L.}~\bibnamefont {Maccone}},\
  }\bibfield  {title} {\bibinfo {title} {Advances in quantum metrology},\
  }\href {https://doi.org/10.1038/nphoton.2011.35} {\bibfield  {journal}
  {\bibinfo  {journal} {Nat. Photon.}\ }\textbf {\bibinfo {volume} {5}},\
  \bibinfo {pages} {222} (\bibinfo {year} {2011})}\BibitemShut {NoStop}%
\bibitem [{\citenamefont {Aaronson}(2018)}]{art:Aaronson}%
  \BibitemOpen
  \bibfield  {author} {\bibinfo {author} {\bibfnamefont {S.}~\bibnamefont
  {Aaronson}},\ }\bibfield  {title} {\bibinfo {title} {Shadow tomography of
  quantum states},\ }in\ \href {https://doi.org/10.1145/3188745.3188802} {\emph
  {\bibinfo {booktitle} {Proceedings of the 50th Annual ACM SIGACT Symposium on
  Theory of Computing}}}\ (\bibinfo {year} {2018})\ p.\ \bibinfo {pages}
  {325–338}\BibitemShut {NoStop}%
\bibitem [{\citenamefont {Huang}\ \emph {et~al.}(2020)\citenamefont {Huang},
  \citenamefont {Kueng},\ and\ \citenamefont {Preskill}}]{art:Huang-Preskill}%
  \BibitemOpen
  \bibfield  {author} {\bibinfo {author} {\bibfnamefont {H.-Y.}\ \bibnamefont
  {Huang}}, \bibinfo {author} {\bibfnamefont {R.}~\bibnamefont {Kueng}},\ and\
  \bibinfo {author} {\bibfnamefont {J.}~\bibnamefont {Preskill}},\ }\bibfield
  {title} {\bibinfo {title} {Predicting many properties of a quantum system
  from very few measurements},\ }\href
  {https://doi.org/10.1038/s41567-020-0932-7} {\bibfield  {journal} {\bibinfo
  {journal} {Nat. Phys.}\ }\textbf {\bibinfo {volume} {16}},\ \bibinfo {pages}
  {1050} (\bibinfo {year} {2020})}\BibitemShut {NoStop}%
\bibitem [{\citenamefont {Wood}\ \emph {et~al.}(2015)\citenamefont {Wood},
  \citenamefont {Biamonte},\ and\ \citenamefont
  {Cory}}]{art:Wood_ChannelRepresentations}%
  \BibitemOpen
  \bibfield  {author} {\bibinfo {author} {\bibfnamefont {C.~J.}\ \bibnamefont
  {Wood}}, \bibinfo {author} {\bibfnamefont {J.~D.}\ \bibnamefont {Biamonte}},\
  and\ \bibinfo {author} {\bibfnamefont {D.~G.}\ \bibnamefont {Cory}},\
  }\bibfield  {title} {\bibinfo {title} {Tensor networks and graphical calculus
  for open quantum systems},\ }\href
  {https://doi.org/https://doi.org/10.26421/QIC15.9-10-3} {\bibfield  {journal}
  {\bibinfo  {journal} {Quantum Inf. Comput.}\ }\textbf {\bibinfo {volume}
  {15}},\ \bibinfo {pages} {759} (\bibinfo {year} {2015})}\BibitemShut
  {NoStop}%
\bibitem [{\citenamefont {Temme}\ \emph {et~al.}(2017)\citenamefont {Temme},
  \citenamefont {Bravyi},\ and\ \citenamefont {Gambetta}}]{art:Temme}%
  \BibitemOpen
  \bibfield  {author} {\bibinfo {author} {\bibfnamefont {K.}~\bibnamefont
  {Temme}}, \bibinfo {author} {\bibfnamefont {S.}~\bibnamefont {Bravyi}},\ and\
  \bibinfo {author} {\bibfnamefont {J.~M.}\ \bibnamefont {Gambetta}},\
  }\bibfield  {title} {\bibinfo {title} {Error mitigation for short-depth
  quantum circuits},\ }\href {https://doi.org/10.1103/PhysRevLett.119.180509}
  {\bibfield  {journal} {\bibinfo  {journal} {Phys. Rev. Lett.}\ }\textbf
  {\bibinfo {volume} {119}},\ \bibinfo {pages} {180509} (\bibinfo {year}
  {2017})}\BibitemShut {NoStop}%
\bibitem [{\citenamefont {Flammia}\ and\ \citenamefont
  {Wallman}(2020)}]{art:Flammia}%
  \BibitemOpen
  \bibfield  {author} {\bibinfo {author} {\bibfnamefont {S.~T.}\ \bibnamefont
  {Flammia}}\ and\ \bibinfo {author} {\bibfnamefont {J.~J.}\ \bibnamefont
  {Wallman}},\ }\bibfield  {title} {\bibinfo {title} {Efficient estimation of
  {P}auli channels},\ }\href {https://doi.org/10.1145/3408039} {\bibfield
  {journal} {\bibinfo  {journal} {{ACM} Trans. Quantum Comput.}\ }\textbf
  {\bibinfo {volume} {1}},\ \bibinfo {pages} {1} (\bibinfo {year}
  {2020})}\BibitemShut {NoStop}%
\bibitem [{\citenamefont {Macchiavello}\ \emph {et~al.}(2004)\citenamefont
  {Macchiavello}, \citenamefont {Palma},\ and\ \citenamefont
  {Virmani}}]{art:QUIT_Pauli}%
  \BibitemOpen
  \bibfield  {author} {\bibinfo {author} {\bibfnamefont {C.}~\bibnamefont
  {Macchiavello}}, \bibinfo {author} {\bibfnamefont {G.~M.}\ \bibnamefont
  {Palma}},\ and\ \bibinfo {author} {\bibfnamefont {S.}~\bibnamefont
  {Virmani}},\ }\bibfield  {title} {\bibinfo {title} {Transition behavior in
  the channel capacity of two-quibit channels with memory},\ }\href
  {https://doi.org/10.1103/PhysRevA.69.010303} {\bibfield  {journal} {\bibinfo
  {journal} {Phys. Rev. A}\ }\textbf {\bibinfo {volume} {69}},\ \bibinfo
  {pages} {010303(R)} (\bibinfo {year} {2004})}\BibitemShut {NoStop}%
\bibitem [{\citenamefont {O'Brien}\ \emph {et~al.}(2004)\citenamefont
  {O'Brien}, \citenamefont {Pryde}, \citenamefont {Gilchrist}, \citenamefont
  {James}, \citenamefont {Langford}, \citenamefont {Ralph},\ and\ \citenamefont
  {White}}]{art:MLE}%
  \BibitemOpen
  \bibfield  {author} {\bibinfo {author} {\bibfnamefont {J.~L.}\ \bibnamefont
  {O'Brien}}, \bibinfo {author} {\bibfnamefont {G.~J.}\ \bibnamefont {Pryde}},
  \bibinfo {author} {\bibfnamefont {A.}~\bibnamefont {Gilchrist}}, \bibinfo
  {author} {\bibfnamefont {D.~F.~V.}\ \bibnamefont {James}}, \bibinfo {author}
  {\bibfnamefont {N.~K.}\ \bibnamefont {Langford}}, \bibinfo {author}
  {\bibfnamefont {T.~C.}\ \bibnamefont {Ralph}},\ and\ \bibinfo {author}
  {\bibfnamefont {A.~G.}\ \bibnamefont {White}},\ }\bibfield  {title} {\bibinfo
  {title} {Quantum process tomography of a controlled-{NOT} gate},\ }\href
  {https://doi.org/10.1103/PhysRevLett.93.080502} {\bibfield  {journal}
  {\bibinfo  {journal} {Phys. Rev. Lett.}\ }\textbf {\bibinfo {volume} {93}},\
  \bibinfo {pages} {080502} (\bibinfo {year} {2004})}\BibitemShut {NoStop}%
\bibitem [{\citenamefont {Smolin}\ \emph {et~al.}(2012)\citenamefont {Smolin},
  \citenamefont {Gambetta},\ and\ \citenamefont {Smith}}]{art:Gambetta_MLE}%
  \BibitemOpen
  \bibfield  {author} {\bibinfo {author} {\bibfnamefont {J.~A.}\ \bibnamefont
  {Smolin}}, \bibinfo {author} {\bibfnamefont {J.~M.}\ \bibnamefont
  {Gambetta}},\ and\ \bibinfo {author} {\bibfnamefont {G.}~\bibnamefont
  {Smith}},\ }\bibfield  {title} {\bibinfo {title} {Efficient method for
  computing the maximum-likelihood quantum state from measurements with
  additive gaussian noise},\ }\href
  {https://doi.org/10.1103/PhysRevLett.108.070502} {\bibfield  {journal}
  {\bibinfo  {journal} {Phys. Rev. Lett.}\ }\textbf {\bibinfo {volume} {108}},\
  \bibinfo {pages} {070502} (\bibinfo {year} {2012})}\BibitemShut {NoStop}%
\bibitem [{\citenamefont {D'Arrigo}\ \emph {et~al.}(2013)\citenamefont
  {D'Arrigo}, \citenamefont {Benenti}, \citenamefont {Falci},\ and\
  \citenamefont {Macchiavello}}]{art:QUIT_FullyCorrelatedDamping}%
  \BibitemOpen
  \bibfield  {author} {\bibinfo {author} {\bibfnamefont {A.}~\bibnamefont
  {D'Arrigo}}, \bibinfo {author} {\bibfnamefont {G.}~\bibnamefont {Benenti}},
  \bibinfo {author} {\bibfnamefont {G.}~\bibnamefont {Falci}},\ and\ \bibinfo
  {author} {\bibfnamefont {C.}~\bibnamefont {Macchiavello}},\ }\bibfield
  {title} {\bibinfo {title} {Classical and quantum capacities of a fully
  correlated amplitude damping channel},\ }\href
  {https://doi.org/10.1103/PhysRevA.88.042337} {\bibfield  {journal} {\bibinfo
  {journal} {Phys. Rev. A}\ }\textbf {\bibinfo {volume} {88}},\ \bibinfo
  {pages} {042337} (\bibinfo {year} {2013})}\BibitemShut {NoStop}%
\bibitem [{\citenamefont {D'Arrigo}\ \emph {et~al.}(2015)\citenamefont
  {D'Arrigo}, \citenamefont {Benenti}, \citenamefont {Falci},\ and\
  \citenamefont {Macchiavello}}]{art:QUIT_CorrelatedDamping}%
  \BibitemOpen
  \bibfield  {author} {\bibinfo {author} {\bibfnamefont {A.}~\bibnamefont
  {D'Arrigo}}, \bibinfo {author} {\bibfnamefont {G.}~\bibnamefont {Benenti}},
  \bibinfo {author} {\bibfnamefont {G.}~\bibnamefont {Falci}},\ and\ \bibinfo
  {author} {\bibfnamefont {C.}~\bibnamefont {Macchiavello}},\ }\bibfield
  {title} {\bibinfo {title} {Information transmission over an amplitude damping
  channel with an arbitrary degree of memory},\ }\href
  {https://doi.org/10.1103/PhysRevA.92.062342} {\bibfield  {journal} {\bibinfo
  {journal} {Phys. Rev. A}\ }\textbf {\bibinfo {volume} {92}},\ \bibinfo
  {pages} {062342} (\bibinfo {year} {2015})}\BibitemShut {NoStop}%
\bibitem [{\citenamefont {Hamada}(2002)}]{art:Hamada}%
  \BibitemOpen
  \bibfield  {author} {\bibinfo {author} {\bibfnamefont {M.}~\bibnamefont
  {Hamada}},\ }\bibfield  {title} {\bibinfo {title} {A lower bound on the
  quantum capacity of channels with correlated errors},\ }\href
  {https://doi.org/10.1063/1.1495537} {\bibfield  {journal} {\bibinfo
  {journal} {J. Math. Phys.}\ }\textbf {\bibinfo {volume} {43}},\ \bibinfo
  {pages} {4382} (\bibinfo {year} {2002})}\BibitemShut {NoStop}%
\end{thebibliography}%
\end{document}